\documentclass[11pt,a4paper, final]{article}
\usepackage{times}
\usepackage[latin1]{inputenc}
\usepackage{amsfonts}
\usepackage{amscd}
\usepackage{amsmath}
\usepackage{amssymb}
\usepackage{theorem}
\usepackage{mathrsfs}
\usepackage[dvips]{graphicx}
\usepackage[english]{babel}
\usepackage{marvosym}
\usepackage{setspace}

\theoremstyle{break}
\theorembodyfont{}
\newtheorem{Def}{Definition}[section]
\newtheorem{rem}[Def]{Remark}
\newtheorem{lem}[Def]{Lemma}
\newtheorem{prop}[Def]{Proposition}
\newtheorem{cor}[Def]{Corollary}
\newtheorem{thm}[Def]{Theorem}
\newtheorem{ex}[Def]{Example}
\newtheorem{nt}[Def]{Notation}

\newcommand*{\rom}[1]{\expandafter\@slowromancap\romannumeral #1@} 
\newcommand{\ir}{r} 
\renewcommand{\phi}{\varphi} 
\newcommand{\E}{\mathbb{E}} 
\newcommand{\R}{\mathbb{R}} 
\newcommand{\N}{\mathbb{N}} 
\newcommand{\mE}{\E} 
\newcommand{\I}{\mathcal J} 

\newcommand{\md}{\;{\rm d}} 
\newcommand{\one}{1\mkern-5mu{\hbox{\rm I}}} 
\makeatletter
\newenvironment{proof}{\noindent{\textit{Proof:}}}{%
\unskip\nobreak\hfil\penalty50\hskip1em\null\nobreak\hfil$\Box$
  \parfillskip=\z@\finalhyphendemerits=0\endgraf\bigskip}
\let\oldendBsp\endBsp
\def\endBsp{\unskip\nobreak\hfil\penalty50\hskip1em\null\nobreak\hfil%
$\blacksquare$\parfillskip=\z@\finalhyphendemerits=0\endgraf\oldendBsp}
\makeatother

\author{Julia Eisenberg \thanks{corresponding author: \small {\tt jeisenbe@fam.tuwien.ac.at}}\qquad 
Paul Kr\"uhner \bigskip\\Institute for Statistics and Mathematical Methods in Economics\smallskip\\TU Wien, Austria}

\date{}

\title{The Impact of Negative Interest Rates on Optimal Capital Injections}

\begin{document}
\maketitle



\begin{abstract}\noindent
In the present paper, we investigate the optimal capital injection behaviour of an insurance company if the interest rate is allowed to become negative. The surplus process of the considered insurance entity is assumed to follow a Brownian motion with drift. The changes in the interest rate are described via a Markov-switching process. It turns out that in times with a positive rate, it is optimal to inject capital only if the company becomes insolvent. However, if the rate is negative it might be optimal to hold a strictly positive reserve. We establish an algorithm for finding the value function and the optimal strategy, which is proved to be of barrier type. Using the iteration argument, we show that the value function solves the Hamilton--Jacobi--Bellman equation, corresponding to the problem. 
\vspace{6pt}

\settowidth\labelwidth{{\it JEL Subject Classification: }}%
                \par\vspace{6pt}           
\settowidth\labelwidth{{\it 2010 Mathematical Subject Classification: }}%
                \par\vspace{6pt}           
\noindent {\it 2010 Mathematical Subject Classification: }%
                \rlap{Primary}\phantom{Secondary}
                93E20;\newline\null\hskip\labelwidth
                Secondary 49L20, 91B30
\end{abstract}

\section{Introduction}
On the 16th of March 2016 the European Central Bank (ECB) set the key interest rate on $0\%$. The deposit facility rate (currently $-0.4\%$) remains negative since the 11th of June 2014, confer \cite{ecb}.
It means, that instead of getting paid for depositing money into the central bank, one has to pay the central bank for it. Also, the yields on government bonds are currently close to their historical minimum. For instance, the yield on the 10-year German government bond, considered one of the safest assets in the world, sank below zero in June 2016 for the first time ever.
\\But why would anyone buy a government bond, lacking annual payments and bringing back less than the amount invested? One reason is the deficit of alternative safe opportunities. Of course, a large corporation can hire guards in order to protect its cash. But doubtless, using bank services is safer and cheaper even in times of negative interest rates. 

Since, insurance companies run massive portfolios of bonds, the changes in the interest rates could be crucial for their balance sheets. Intuitively, it is clear that ultra-low interest rates immensely affect the life insurance sector: the long-term promises to policyholders, made decades ago, imply a much higher interest rate and cause mismatches between assets and liabilities.   
\\But do negative interest rates affect the value of a non-life insurance company?
Typically, one assumes that there is little impact because most policies are short-termed, implying that the assets and liabilities can be properly matched. However, this perspective neglects the value of future business potential, for instance future premia (competitive  markets), dividends (profitability) or capital injections (Solvency II capital requirements).
\\Indeed, non-life insurance premia should be based on the premise of appropriate pricing and give a ``forecast'' on profitability and possible dividend payments. 
Therefore, the premia are highly dependent on the economic markers. 
Also, Solvency II emphasizes the importance of incorporating all the risks, including the inflation risk and the interest rate risk, for the calculation of the capital requirement.

The crisis of 2008 and the bad situation in 2015, which is considered as the worst year since the crisis of 2008, let the economists speak of business-cycle dynamics characterized by more than one interest rate, confer for instance \cite{gard}. Mathematically one can translate the cycle dynamics into a Markov-switching model, where the interest rate switches on random times and is kept constant inbetween. This model has been widely investigated in the mathematical finance literature, confer for instance Boyarchenko and Levendorskii \cite{bl}, Jiang and Pistorius \cite{jp} or Duan et al. \cite{duan}. In actuarial mathematics, some recent results on the risk theory in a Markovian environment can be found for instance in Asmussen \cite{asmus} or B\"auerle \cite{bae}, some optimisation problems have been investigated for example in Zhu and Yang \cite{zhu} or Jiang and Pistorius \cite{jp2}. 

Throughout the life cycle of a business, a company can face considerable economic challenges and multiple instances of financial distress. As a consequence, it might require capital injections to remain afloat. In actuarial mathematics, the term capital injections and the corresponding risk measure have been proposed in the discussion in Pafumi \cite{paf}. Further discussions can be found in Dickson and Waters \cite{dick1}, Eisenberg and Schmidli \cite{es} or in Nie et al. \cite{nie}. In their study Nie et al. even assume that the capital injections do not eliminate  the possibility of ruin for the insurer.

In the present paper, we assume that the considered insurance entity models its surplus via a Brownian motion with drift. The interest rate can attain a negative and a positive value, mimicking a business-cycle with two states. The target is to minimise the value of expected discounted capital injections, under the constraint that the company is not allowed to become insolvent. It is intuitively clear that in the time periods with positive interest rates, it is optimal to inject capital just if the surplus becomes negative and just as much as is necessary to land at zero. However, in times with negative yields it might be optimal to hold a strictly positive reserve. The heuristic explanation is that early injections appear cheaper than later payments. 
\\The paper is organised as follows. In Section 2, we formulate the problem and investigate its well-posedness. In Section 3, we briefly consider the strategy with minimal-amount injections, identify the optimal strategy as a barrier strategy and
introduce an algorithm for approximation of the value function.

\section{Model Setup}
Consider an insurance company whose surplus is given by a Brownian motion with drift $X_t=x+\mu t+\sigma W_t$, where $W$ is a standard Brownian motion $\mu,\sigma>0$. We assume that the underlying filtration $\mathcal F$ is complete, right-continuous and that $W$ is a standard $\mathcal F$-Brownian motion.
Further, we model the stochastic interest rate $\ir$ as a continuous time $\mathcal F$-Markov chain. For simplicity, we assume that the state space $\mathcal S$ consists of only two points $\delta_1\leq 0<\delta_2$ and the Markov chain switches with intensities $\lambda_1,\lambda_2>0$ respectively. 

The insurance company is allowed to ask for capital injections at any time, where the accumulated capital injections until $t$ are given by $Y_t$, yielding for the ex-controlled surplus $X^{Y}$:
\[
X_t^Y=x+\mu t+\sigma W_t+Y_t\;.
\]
We call a strategy $Y$ admissible if $Y$ is a right-continuous, non-decreasing and $\mathcal F$-adapted process which starts in zero with $Y_t\geq (-\inf\{X_s:s\in[0,t]\})\vee 0$. We denote the class of those processes by $\mathcal A$.

As a risk measure, we consider the value of expected discounted injections, where the injected capital is discounted by the stochastic interest rate $\ir_t$. The return function corresponding to an admissible strategy $Y\in\mathcal A$ is given by:
$$ V^Y(x,\eta) := \mE_{x,\eta}\Big[\int_0^\infty e^{-\int_0^t\ir_s\md s}\md Y(s)\Big]\;,$$
where the indices $x$ and $\eta$ indicate $X_0=x$ and $\ir_0=\eta$.
We seek to minimise the total discounted injected capital, i.e.\ we seek to find an admissible strategy $Y^*$ such that
 \begin{align}\label{e:cp}
 V(x,\eta):=\inf\limits_{Y\in\mathcal A} V^Y(x,\eta) = V^{Y^*}(x,\eta)\;,\quad x\geq 0, \eta\in \mathcal S\;.
 \end{align}
The formal corresponding Hamilton--Jacobi--Bellman equation for $i,j\in\{1,2\}$, $i\neq j$ and $x\ge 0$ is
\begin{align}\label{HJB1}
\begin{split}
  \min\Big\{ \frac{\sigma^2}{2} V''(x,\delta_i) + \mu V'(x,\delta_i) - (\delta_i+\lambda_i)V(x,\delta_i) + \lambda_iV(x,\delta_j),
	\\V'(x,\delta_i) + 1 \Big\} = 0\;.
\end{split}
\end{align}
\begin{nt}
For the sake of convenience, we introduce the following notation\medskip
\\$\bullet$
\begin{align*}
\mathcal L_i(f)(x)=\frac{\sigma^2}{2} f''(x) + \mu f'(x) - (\delta_i+\lambda_i)f(x) 
\end{align*}
for $i\in\{1,2\}$ and a sufficiently smooth function $f$. $\mathcal L_i$ can be applied also on $V^Y(x,\eta)$, whereas the notation $(V^Y)'(x,\eta)$ denotes the derivative with respect to $x$.
\medskip
\\$\bullet$ We define $Y_t^0:=0\vee -\inf\{X_s:s\in[0,t]\}$, the corresponding return function and the ex-injection process will be denoted by $V^0(x,\eta)$ and by $X^{0}:=X^{Y^0}$ respectively. In the following we call the strategy $Y^0$ the minimal-amount strategy.
\end{nt}
Since a negative interest rate can lead to an infinite return function, we have to find the conditions under which the minimisation problem is well-posed. That is, we want to find an admissible strategy $Y$ such that $V^Y(x,\eta)<\infty$ for $x\ge 0$, $\eta\in \mathcal S$.
\begin{prop}\label{p:well posed}
Assume that $\delta_1>-\frac{\lambda_1\delta_2}{\lambda_2+\delta_2}$. Then, the strategy $Y^0$ satisfies $$V^{0}(x,\eta)<\infty$$
for any $x\ge 0$, $\eta\in \mathcal S$. In particular, the stochastic control problem \eqref{e:cp} is well-posed.
\end{prop}
\begin{proof}
Let $x\geq 0$ and $\eta\in \mathcal S$. Clearly, the strategy $Y^0$ is independent of the stochastic interest rate process $\ir$. First we calculate the average interest rate and then we relate it to the expectation. Define the occupation time of the stochastic interest rate in the level $\delta_1$ by $\Lambda(t):=\int_0^t\one_{\{\ir_s=\delta_1\}}\md s$ for any $t\geq 0$. Then, we have $\int_0^t\ir_s\md s = t\delta_2 + (\delta_1-\delta_2)\Lambda(t)$ for $t\geq 0$. Hence, we get
$$ \E_{x,\eta}\Big[\exp\Big(-\int_0^t\ir_s\md s\Big)\Big] = \exp(-t\delta_2)\E_{x,\eta}\Big[\exp\Big(-(\delta_1-\delta_2)\Lambda(t)\Big)\Big],\quad t\geq 0.$$
Let
$$I:=\left(\begin{matrix}1&0\\0&1\end{matrix}\right)\quad \mbox{and}\quad R := \left(\begin{matrix}-\lambda_1-\delta_1+\delta_2&\lambda_1 \\\lambda_2&-\lambda_2\end{matrix}\right)\;.$$
From \cite[p. 385]{pedler.71} one knows
$$ \E_{x,\eta}\Big[\exp\Big(-(\delta_1-\delta_2)\Lambda(t)\Big)\Big] =  \left(\one_{\{\ir_0=\delta_1\}},\one_{\{\ir_0=\delta_2\}}\right)\cdot\exp(tR)\cdot
  \left(\begin{matrix}1\\1\end{matrix}\right),\quad t\geq 0\;.$$ 
Defining 
\begin{align*}
&a:=-(\lambda_1+\lambda_2+\delta_1+\delta_2)\;, \\&b:=\sqrt{(\lambda_1+\lambda_2+\delta_1-\delta_2)^2+4\lambda_2(\delta_2-\delta_1)}\;,
\\&\omega_1:= \delta_2+\frac{1}{2}\left(a+b\right)\;,\quad \omega_2 := \delta_2+\frac{1}{2}\left(a-b\right),
  \end{align*}
we find that
  $$ \exp(tR) = \frac{\omega_1e^{t\omega_2}-\omega_2e^{t\omega_1}}{\omega_1-\omega_2}\cdot I + \frac{e^{t\omega_1}-e^{t\omega_2}}{\omega_1-\omega_2} \cdot R\;.$$ Since, 
  \begin{align*}
    \left(\one_{\{\ir_0=\delta_1\}},\one_{\{\ir_0=\delta_2\}}\right)\cdot I \cdot
  \left(\begin{matrix}1\\1\end{matrix}\right) &= 1\quad\text{and} \\
    \left(\one_{\{\ir_0=\delta_1\}},\one_{\{\ir_0=\delta_2\}}\right)\cdot R\cdot
  \left(\begin{matrix}1\\1\end{matrix}\right) &= (\delta_2-\delta_1)\one_{\{\ir_0=\delta_1\}}
  \end{align*}
we find that
\begin{align*} 
\mE_{x,\eta}\left[\exp\left(-\int_0^t\ir_s\md s\right)\right] &= \frac{e^{t(\omega_1-\delta_2)}-e^{t(\omega_2-\delta_2)}}{\omega_1-\omega_2}(\delta_2-\delta_1)\one_{\{\ir_0=\delta_1\}}
\\&\quad {}+\frac{\omega_1e^{t(\omega_2-\delta_2)}-\omega_2e^{t(\omega_1-\delta_2)}}{\omega_1-\omega_2}\;.
\end{align*}
Observe that $\omega_2-\delta_2<\omega_1-\delta_2=\frac{1}{2}(a+b)=:-c<0$ by assumption. Hence, there is a positive constant $C>0$ depending on $\lambda_1,\lambda_2,\delta_1,\delta_2$ such that
  $$ \E_{x,\eta}\left[\exp\left(-
  \int_0^t\ir_s \md s\right)\right] \leq C\exp(-tc),\quad t\geq 0.$$
Thus, we have
$$ V^{0}(x,\eta) \le C \E_{x,\eta}\left[\int_0^\infty e^{-cs}\md Y^0_s\right]<\infty. $$
\end{proof}
\section{The Value Function and the Optimal Strategy}
In this section we aim at identifying the value function and the optimal strategy. From now on, we always assume \bigskip
\\\textbf{Assumption:} $\delta_1>-\frac{\lambda_1\delta_2}{\lambda_2+\delta_2}>-\lambda_1$. 
\bigskip
\\Then, Proposition \ref{p:well posed} yields that the stochastic control problem \eqref{e:cp} is well-posed.
\subsection{Performance of the minimal-amount injection strategy}
We start our investigation by analysing the performance of the minimal-amount injection strategy $Y^0$, which turns out to be optimal in some cases. We calculate its performance function $V^0(x,\eta)$ in Proposition \ref{prop:y0} below. There, we also specify the conditions under which $Y^0$ is the optimal injection strategy. 
\begin{prop}\label{prop:y0}
 For $\lambda_2> 0$ define 
\begin{align*}
 &a:=\lambda_1+\delta_1+\lambda_2+\delta_2&&\mbox{and} &&\alpha:= \lambda_1+\delta_1-\lambda_2-\delta_2,\\
    &D_1 := \frac {a  - \sqrt{ \alpha^2 +4\lambda_1\lambda_2}}2 &&\mbox{and}   
    &&D_2 := \frac {a  + \sqrt{ \alpha^2 +4\lambda_1\lambda_2}}2, \\
&A_1 := \frac{\mu+\sqrt{\mu^2+2\sigma^2D_1}}{\sigma^2} &&\mbox{and}
&&A_2 := \frac{\mu+\sqrt{\mu^2+2\sigma^2D_2}}{\sigma^2}, \\
&E:= \frac{\lambda_2+\delta_2-D_1}{\lambda_2} && \mbox{and}
&&F:= \frac{\lambda_2+\delta_2-D_2}{\lambda_2}, \\
&B_2:= \frac{1-F}{A_1(E-F)} && \mbox{and}
&&C_2:= \frac{E-1}{A_2(E-F)}, \\
&B_1:= EB_2 && \mbox{and}
&&C_1:= FC_2. \\
\end{align*}
Then we have
\begin{align*}
V^0(x,\delta_1) &= B_1e^{-A_1\cdot x} +   C_1e^{-A_2\cdot x}, \\
V^0(x,\delta_2) &= B_2e^{-A_1\cdot x} +   C_2e^{-A_2\cdot x}
\end{align*} 
for any $x\geq 0$. Moreover, $V^0 = V$ if and only if $B_1A_1^2 + C_1A_2^2\geq 0$. In this case $Y^0$ is the optimal injection strategy.\medskip
\\If $\lambda_2=0$, the calculations become much simpler. In this one knows immediately 
\[
V^0(x,\delta_2)=\frac{\sigma^2}{\mu+\sqrt{\mu^2+2\sigma^2(\delta_2)}}e^{-\frac{\mu+\sqrt{\mu^2+2\sigma^2(\delta_2)}}{\sigma^2}x}\;.
\]
$V^0(x,\delta_1)$ can be easily obtained via solving the differential equation
\[
\mathcal L_1(V^0)(x,\delta_1)+\lambda_1 V^0(x,\delta_2)=0
\]
with boundary conditions $(V^0)'(0,\delta_1)=-1$ and $\lim\limits_{x\to\infty}V^0(x,\delta_1)=0$.
\end{prop}
\begin{proof}
Due to the assumption on $\delta_1$, we have $\delta_1\delta_2+\delta_1\lambda_2+\lambda_1\delta_2>0$ and, hence, $D_2>D_1>0$. Also, we see that $A_1,A_2>0$ and $E>F$.
\\  
Additionally, we have $D_j = \frac{\sigma^2}2A_j^2+\mu A_j$ for $j\in\{1,2\}$. Now, it is easy to see that for $i,j\in\{1,2\}$ with $i\neq j$ it holds
 \begin{align*}
\mathcal L_i(V^0)(x,\delta_i) + \lambda_i V^0(x,\delta_j) = 0\;.
 \end{align*}
 and the right-hand side of the claimed equality is the unique solution to these systems of ODEs with derivative $-1$ in $x=0$ and vanishing at infinity. Thus, we have
\begin{align*}
  V^0(x,\delta_1) &= B_1e^{-A_1\cdot x} +   C_1e^{-A_2\cdot x}, \\
  V^0(x,\delta_2) &= B_2e^{-A_1\cdot x} +   C_2e^{-A_2\cdot x}
\end{align*} 
for any $x\geq 0$. Also, $V^0(\cdot,\delta_2)$ is convex and, hence, $(V^0)'(\cdot,\delta_2)\geq -1$ which yields
 $$ \min\{\mathcal L_2(V^0)(x,\delta_2)+\lambda_2 V^0(x,\delta_1), (V^0)'(x,\delta_2) + 1\} = 0$$
for any $x\geq 0$. We see that $V^0$ is a $\mathcal C^2$-function and solves the HJB equation \eqref{HJB1} iff $(V^0)'(x,\delta_2) \geq -1$ for any $x\geq 0$. 
\medskip
\\However, $(V^0)''(x,\delta_2)$ has at most one zero $x_0\geq 0$ because it is the sum of two exponential functions. Above this zero we must have $(V^0)''(x,\delta_2)\geq 0$ because $V^0$ is decreasing. Consequently, $(V^0)''(x,\delta_2)<0$ on $[0,x_0]$ if such a zero $x_0$ exists.
\medskip
\\Now, if $B_2A_1^2+C_2A_2^2\geq 0$, then $(V^0)''(0,\delta_2) \geq 0$ and, hence, we either have $x_0=0$ or $(V^0)''(x,\delta_2)$ does not have any zeros. Hence, $V^0(x,\delta_2)$ is convex and, thus, we have $(V^0)'(x,\delta_2)\geq -1$.
\medskip
\\If $(V^0)'(x,\delta_2) \geq -1$ for any $x\geq 0$, then $0\leq (V^0)''(0,\delta_2) = B_2A_1^2+C_2A_2^2$ as claimed.
\end{proof}

\subsection{Recursion}
One might ask why it is necessary to establish a recursion if one can tackle the problem by solving the corresponding differential equation. The problem lies in the correct choice of the optimal barrier level. It turns out that the function to minimise exhibits a complex non-linear dependence on the barrier $b$ as a variable. Even in this two states problem it is a hard challenge to find the optimal barrier in the negative state. The complexity of the problem increases significantly with the number of states. In contrast, the recursion could be generalised to an arbitrary number of states.

In this section we construct a sequence of functions $(V_n)_{n\in \mathbb N}$ such that $V_{2n} \rightarrow V(\cdot,\delta_2)$ and $V_{2n+1} \rightarrow V(\cdot,\delta_1)$ uniformly together with their first two derivatives. The function $V_n$ is actually the value function of the following modified problem: The same as the original problem but we start in $\delta_1$ if $n$ is odd and in $\delta_2$ if it is even, and no more capital injections need to be made after the $n+1$ change in the interest rate $\ir$.

Obviously, we have to invest less in the modified problems and thus we expect that $V_n\leq V$. The optimal strategy in the modified problems are proved to be of barrier type, where the barriers are adjusted at the switching times of the interest rate.
\subsubsection{Initial step:}
Consider at first the auxiliary problem where we seek to minimise the value of expected discounted capital injections for the preference rate $\delta_2>0$ up to an exponentially distributed stopping time $T_2\sim{\rm Exp}(\lambda_2)$. Because $\delta_2>0$, it is immediately clear that the optimal barrier is given by $0$, i.e. the optimal strategy $Y^0$ is to inject capital just in the case the surplus becomes negative and just as much as to shift the process back to zero. Since $Y^0$ and $T_2$ are independent, we obtain 
\begin{align*}
\mE_x\Big[\int_0^{T_2}e^{-\delta_2 t}\md Y_t^0\Big]=\mE_x\Big[\int_0^\infty e^{-(\delta_2+\lambda_2)t}\md Y^0_t\Big]\;.
\end{align*}
Therefore, compare for instance \cite{es}, the value function is given by
\[
V_0(x):=\frac1{A_2}e^{-A_2 x},\quad \quad A_2:=\frac{\mu+\sqrt{\mu^2+2\sigma^2(\delta_2+\lambda_2)}}{\sigma^2}\;,
\]
i.e. $V_0(x) = \inf\limits_{Y\in\mathcal A} \E_x[\int_0^{T_2}e^{-\delta_2t}\md Y_t]$ for $x\ge 0$. 
\begin{rem}
Analogously, if we merely have $\delta_1>-\lambda_1$, then we could have done the same approach starting from the negative interest-rate state except that the constant $A_2$ has to be replaced by 
\[
A_1:=\frac{\mu+\sqrt{\mu^2+2\sigma^2(\delta_1+\lambda_1)}}{\sigma^2}\;.
\]
For the sake of convenience, we additionally define
\[
\tilde A_1:=\frac{-\mu+\sqrt{\mu^2+2\sigma^2(\delta_1+\lambda_1)}}{\sigma^2}\;.
\]
\end{rem}
\subsubsection{Further steps:}
Analogously to Initial step, we denote by $V_n$ the value function of the problem with $n$ jumps, where after the $n$th jump one lands in the state with $\delta_2>0$ and stops the consideration at the next exponential switching time.
In the following, we construct the value functions $(V_n)_{n\in\mathbb N}$ along with the optimal barriers $b_n$. Proposition \ref{p:ODE} points out that our definitions do actually make sense and Lemma \ref{l:V_n solves HJB} verifies that $V_n$ is indeed the value function of the modified problem for every $n\in\N$. 
Due to the construction of our auxiliary problems, it is clear that in the $(2n)$th problem we start with the $\delta_2>0$ state, and in the $(2n+1)$st problem with the $\delta_1\le 0$ state. Theorem \ref{t:convergence result} states that the sequences $(V_{2n})_{n\ge 1}$ and $(V_{2n-1})_{n\ge 1}$ converge to the value function $V$ of the original problem in a suitable way.

It is clear, that in times of positive interest rate, it is optimal to inject as late and as less as possible. That, is we know the optimal strategy: it is a barrier strategy with barrier $b_{2n}:=0$. Then, knowing the value function of the $(2n-1)$st problem, we can easily calculate the value function of the $(2n)$th problem. During times of negative interest rate, it is cheaper to inject early at once, but the optimal amount is not obvious. If the optimal strategy is a constant barrier strategy, then this barrier is independent of the surplus level. In order to simplify the calculations, we can start by finding the optimal barrier for zero initial surplus. Imagine now, we have already calculated the value function of the $(2n)$th problem. We optimise the level of the barrier $b\ge 0$ until the next switching time $T_1\sim{\rm Exp}(\lambda_1)$. The return function $V^b$, corresponding to the strategy: keep the surplus over $b$ up to $T_1$ and then follow the optimal strategy from $2n$,  yields
\[
V^b(0)=\mE_0\Big[\int_0^{T_1} e^{-\delta_1 s}\md Y_s^0+b+e^{-\delta_1 T_1}V_{2n}\big(b+X_{T_1}^0\big)\Big]\;.
\]
In order to find a $b$ minimising the above function, we have to consider just the terms depending on $b$:
$$g_{2n}:b\mapsto b + \E_0\left[e^{-\delta_1 T_1} V_{2n}\Big(b + X^{Y^0}_{T_1}\Big) \right].$$
Due to Corollary \ref{k:I=N} below, $V_{2n}$ is strictly decreasing and convex, which means that $g_{2n}$ has a unique minimum. We choose recursively a minimum $b_{2n+1}$ for the function $g_{2n}$ and define recursively $V_{2n+1}$ as the unique solution to the ODE
 \begin{align}\label{eq:2n+1}
 \frac{\sigma^2}{2}V''_{2n+1}(x) + \mu V'_{2n+1}(x) -(\delta_1+\lambda_1)V_{2n+1}(x) + \lambda_1 V_{2n}(x) = 0
\end{align}  
for $x\geq b_{2n+1}$ with $V_{2n+1}'(b_{2n+1})=-1$, $\lim\limits_{x\rightarrow\infty} V_{2n+1}(x) = 0$ and
 $$ V_{2n+1}(x) := V_{2n+1}(b_{n+1}) + (b_{2n+1}-x), \quad x\in[0,b). $$
Also, we define $V_{2n+2}$ as the unique solution to the ODE
 \begin{align}\label{eq:2n} \frac{\sigma^2}{2}V''_{2n+2}(x) + \mu V'_{2n+2}(x) -(\delta_2+\lambda_2)V_{2n+2}(x) + \lambda_2 V_{2n+1}(x) = 0
 \end{align}
for $x\geq 0=b_{2n}$ with $V_{2n+2}'(0)=-1$ and $\lim\limits_{x\rightarrow\infty} V_{2n+2}(x) = 0$.

As we will see, $(V_n)_{n\in\mathbb N}$ defines a sequence of convex, decreasing $\mathcal C^2$-functions vanishing together with their derivatives at infinity. 

Let 
\begin{align*}
\I := \{n\in\mathbb N: V_n\in \mathcal C^2,\; V_n>0,\;V_n''\ge 0,\; V_n'<0\;,\lim_{x\rightarrow\infty}V_n(x)=0\}
\end{align*}
and note that $0\in \I$. Corollary \ref{k:I=N} below implies that $\I=\mathbb N$. Then, for $n\in\mathbb N$ it holds $g_{2n}'(b) = 1 + \E_0\left[e^{-\delta_1 T_1} V'_{2n}(b + X^0_{T_1}) \right]$ with a unique zero $b_{2n+1}$ which satisfies $b_{2n+1}=0$ if $g'_{2n}(0)\geq 0$ or
 $$ -1 = \E_0\left[e^{-\delta_1 T_1} V'_{2n}(b_{2n+1} + X^0_{T_1}) \right]. $$
 
Next we will show that $V_{n+1}$ is, indeed, twice continuously differentiable for any $n\in \I$. Since $V_{n+1}$ solves the ODE \eqref{eq:2n} or \eqref{eq:2n+1} on $[b_{n+1},\infty)$ and since it is linear below $b_{2n+1}$ with slope $-1$ it is clear that it is a $\mathcal C^1$-function which is twice continuously differentiable on $\mathbb R_+\backslash \{b_{n+1}\}$. If $b_{n+1}=0$, then $V_{n+1}$ is twice continuously differentiable. If $b_{n+1}>0$, then the second left-side derivative in $b_{n+1}$ equals zero because $V_{n+1}$ is linear below $b_{n+1}$. The next lemma observes that with our choice of $b_{n+1}$ the right-side derivative vanishes as well if $b_{n+1}>0$.
\begin{lem}\label{l:ODE vanish}
  Let $2n\in \I$. If $b_{2n+1}>0$, then $V_{2n+1}''(b_{2n+1}) = 0$. If $b_{2n+1}=0$, then $V_{2n+1}''(b_{2n+1})=V_{2n+1}''(0)\geq0$.
In particular, $V_{2n+1}$ is twice continuously differentiable.
\end{lem}
\begin{proof}
Assume first that $b_{2n+1}>0$ and let $T$ be an Exp($\lambda_1+\delta_1$)-distributed random variable which is independent of $(X,Y^0)$. Then, we have $g_{2n}'(b_{2n+1})=0$ and, hence,
  \begin{align*}
    -1 &= \E_0[e^{-\delta_1 T_1}V'_{2n}(b_{2n+1}+X^{0}_{T_1})] \\
      &= \frac{\lambda_1}{\lambda_1+\delta_1}\int_0^\infty (\lambda_1+\delta_1)e^{-t(\lambda_1+\delta_1)}\E_0[V'_{2n}(b_{2n+1}+X^{0}_{t})]\md t \\
     &= \frac{\lambda_1}{\lambda_1+\delta_1}\E_0[V'_{2n}(b_{2n+1}+X^{0}_{T})] \\
     &= \frac{\lambda_1}{\lambda_1+\delta_1} \int_0^\infty V'_{2n}(b_{2n+1}+y)\frac{2(\lambda_1+\delta_1)}{\sigma^2A_1}e^{-y\tilde A_1}\md y \\
     &= \frac{2\lambda_1}{\sigma^2A_1}\left(\int_0^\infty V_{2n}(b_{2n+1}+y)\tilde A_1e^{-\tilde A_1 y}\md y-V_{2n}(b_{2n+1})\right)
  \end{align*}
  where we used that the density of $X^{0}_{T_1}$ is $\rho(y)=\frac{2(\lambda_1+\delta_1)}{\sigma^2A_1}e^{-\tilde A_1y}$, $y\ge 0$ given in Borodin and Salminen \cite[p. 252]{bs}, Formula 1.2.6. Thus, we get
   $$ V_{2n}(b_{2n+1}) = \frac{\sigma^2A_1}{2\lambda_1}+\tilde A_1\int_{b_{2n+1}}^\infty V_{2n}(z)e^{(b_{2n+1}-z)\tilde A_1}\md z\;.$$
Rewriting the ODE \eqref{eq:2n+1} and inserting for $V_{2n+1}(b_{2n+1})$ the value given in \eqref{eq:boundary}, calculated in Proposition \ref{p:ODE}, yields
  \begin{align*}
    \frac{\sigma^2}{2}V_{2n+1}''(b_{2n+1}) &= \mu + (\lambda_1+\delta_1)V_{2n+1}(b_{2n+1}) - \lambda_1 V_{2n}(b_{2n+1}) \\
    &= 0.
  \end{align*}
Now assume that $b_{2n+1}=0$. Then $g_{2n}$ attains its minimum in $0$ and $g_{2n}'(b_{2n+1})\geq 0$. Thus, we have
  \begin{align*}
    -1 &\leq \E_0[e^{-\delta_1 T_1}V'_{2n}(X^{0}_{T_1})] \\
       &= \frac{2\lambda_1}{\sigma^2A_1}\left(\int_0^\infty V_{2n}(y)(\tilde A_1)e^{-y\tilde A_1}\md y-V_{2n}(0)\right)
  \end{align*}
  which implies that
  $$ V_{2n}(0) \leq \frac{\sigma^2A_1}{2\lambda_1}+\tilde A_1\int_{b_{2n+1}}^\infty V_{2n}(z)e^{(b_{2n+1}-z)\tilde A_1}\md z.$$
  Hence, we get
   $$ V_{2n+1}''(0) \geq 0. $$
\end{proof}
Finally, we find that $\I=\N$ and, hence, $V_n$ is a convex, twice continuously differentiable, decreasing and positive valued function.
\begin{cor}\label{k:I=N}
It holds $\I = \N$.
\end{cor}
\begin{proof}
  Let $k\in \I$.
  
  \underline{Case 1}: $k$ is even. Then there is an $n\in \mathbb N$ such that $k=2n$. Lemma \ref{l:ODE vanish} together with Proposition \ref{p:ODE} yield that $V_{2n+1}$ is a convex twice continuously differentiable function and, hence, $k+1=2n+1\in \I$.
  
  \underline{Case 2}: $k$ is odd. Then there is $n\in\mathbb N$ such that $k=2n+1$. Proposition \ref{p:ODE} yields that $k+1 = 2n+2\in\I$.
\\Since $0\in \I$ we get $\I=\mathbb N$.
\end{proof}
With the preceding at hand we can now prove the (pointwise) monotonicity of the sequences $(V_{2n})_{n\in\mathbb N}$, $(V_{2n+1})_{n\in\mathbb N}$ and $(b_{2n+1})_{n\in\mathbb N}$.
\begin{lem}\label{l:monotonie}
For any $n\in\mathbb N$, $x\geq 0$ we have $V_{n+2}(x)\geq V_{n}(x)$ and we have $b_{n+2} \leq b_n$.
\end{lem}
\begin{proof}
  Let $\I_2 := \{n\in \mathbb N: b_{n+2} \leq b_n,\forall x\geq0 : V_{n+2}(x)\geq V_{2n}(x)\}$.
\\We show that $0\in\I_2$. Simply observe that
   \begin{align*}
      &\frac{\sigma^2}{2}V_{0}''+\mu V_0' -(\lambda_2+\delta_2) V_0 = 0, \\
      &\frac{\sigma^2}{2}V_{2}''+\mu V_2' -(\lambda_2+\delta_2) V_2 + \lambda_2 V_1 = 0, 
   \end{align*}
$V_1$ is strictly positive, $V_2'(0)=-1=V_0'(0)$ and $V_2(0) > V_0(0)$. Hence, \cite{walter.98} yields that $V_2(x) > V_0(x)$ for any $x\geq 0$. Consequently, $0\in \I_2$. 
\\Now let $n\in \I_2$.
   
  \underline{Case 1:} $n$ is odd. Then $n+1$ is even and, hence, we have
   \begin{align*}
      &\frac{\sigma^2}{2}V_{n+1}''+\mu V_{n+1}' -(\lambda_2+\delta_2) V_{n+1} +\lambda_2 V_{n} = 0, \\
      &\frac{\sigma^2}{2}V_{n+3}''+\mu V_{n+3}' -(\lambda_2+\delta_2) V_{n+3} + \lambda_2 V_{n+2} = 0, 
   \end{align*}
$V_{n+2}\geq V_{n}$, $V_{n+3}'(0)=-1=V_{n+1}'(0)$ and $V_{n+3}(0) > V_{n+1}(0)$. Hence, \cite{walter.98} yields that $V_{n+3}(x) > V_{n+1}(x)$ for any $x\geq 0$. Since $b_{n+3}=0=b_{n+1}$ we have $n+1\in \I_2$. 
   
  \underline{Case 2:} $n$ is even. Then $n+1$ is odd. Since $n\in \I_2$ we get $b_{n+3}\leq b_{n+1}$. Let $W_{n+1}$ be the solution to the ODE
   $$ \frac{\sigma^{2}}{2}W_{n+1}'' + \mu W_{n+1}' -(\lambda_1+\delta_1)W_{n+1} + \lambda_1 V_{n} = 0$$
   with $\lim\limits_{x\rightarrow\infty} W_{n+1}(x) = 0$ and $W_{n+1}'(b_{n+3})=-1$. Then $W_{n+1}\geq V_{n+1}$ on $[b_{n+3},\infty)$. Also, the comparison principle \cite{walter.98} yields that $V_{n+3}\geq W_{n+3}$ on $[b_{n+3},\infty)$ and, hence, $V_{n+3} \geq V_{n+1}$ on $[b_{n+3},\infty)$. Since $V_{n+3}$, $V_{n+1}$ are linear with slope $-1$ on $[0,b_{n+3}]$ we get that $V_{n+3}\geq V_{n+1}$ on $\mathbb R_+$. Thus, we have $n+1\in \I_2$.
\medskip
\\Consequently, $\I_2 = \mathbb N$ which is the claim.
\end{proof}
With all the properties at hand we can show that $V_n$ is the value function of the modified control problem introduced at the beginning of the section.
\begin{lem}\label{l:V_n solves HJB}
We have
\begin{align*} 
\min\Big\{\mathcal L_j(V_{2n+j})(x)+\lambda_j V_{2n+j-1}(x),V_{2n+j}'(x)+1\Big\} = 0
\end{align*}
for any $n\in\mathbb N$, $x\geq0$, $j\in\{1,2\}$. 
In other words we have
\begin{align*}   
V_{2n+j}(x) = \inf\limits_{Y\in\mathcal A} \E_x\left[\int_0^{T_j}e^{-\delta_j s}\md Y_s+e^{-\delta_j T_j}V_{2n+j-1}(X^Y_{T_j})\right]
\end{align*}
where $T_j$ is an $(X,Y)$-independent ${\rm Exp}(\lambda_j)$-distributed random variable.
\end{lem}
\begin{proof}
Let $n\in\mathbb N$, $x\geq0$ and $j\in\{1,2\}$. 
\\If $j=2$, then $V_{2n+j}'+1\geq 0$ and
  $$ \frac{\sigma^2}{2}V_{2n+j}''(x) + \mu V_{2n+j}'(x) -(\lambda_j+\delta_j) V_{2n+j}(x)+\lambda_j V_{2n+j-1}(x) = 0\;.$$
Therefore, the claim holds. Hence, we may assume that $j=1$. Recall that $V_{2n+1}$ solves the differential equation \eqref{eq:2n+1} for $x\in[b_{2n+1},\infty)$ and fulfils $V_{2n+1}(x)=V_{2n+1}(b_{2n+1})+b_{2n+1}-x$ for $x\in[0,b_{2n+1})$. 
\\If $x\geq b_{2n+1}$, then we can prove the claim like described in the first case. 

Now, assume by contradiction that there is $0\le x_0<b_{2n+1}$ such that
 $$ \frac{\sigma^2}{2}V_{2n+1}''(x_0) +\mu V_{2n+1}'(x_0)-(\delta_1+\lambda_1)V_{2n+1}(x_0) +\lambda_1 V_{2n}(x_0) < 0 $$
Let $\widetilde V$ be the solution to the ODE
 $$ \frac{\sigma^2}{2}\widetilde V''(x) +\mu \widetilde V'(x)-(\delta_1+\lambda_1)\widetilde V(x) +\lambda_1 V_{2n}(x) = 0$$
for $x\in[x_0,\infty)$ with $\widetilde V'(x_0) = -1$ and $\lim\limits_{x\rightarrow\infty}\widetilde V(x) = 0$, cf.\ Proposition \ref{p:ODE}. We also define $\tilde V(x) := \tilde V(x_0) + (x_0-x)$ for $x\in [0,x_0)$. Since $x_0< b_{2n+1}$ Corollary \ref{c:2nd.abl.wachs} yields $\tilde V''(x_0) \leq V''_{2n+1}(b_{2n+1}) = 0$ and the latter equality holds by Lemma \ref{l:ODE vanish}. $\tilde V$ is the performance function of the strategy with barrier $x_0$ until time $T_1$ and following the optimal strategy afterwards. $b_{2n+1}$ is chosen such that the expected discounted capital injections are minimised among barrier strategies if the initial capital is zero, i.e.\ $\tilde V(0) \geq V_{2n+1}(0)$. Thus, we get $\tilde V(x_0) \geq \tilde V_{2n+1}(x_0)$ by linearity with slope $-1$. Then, we have
\begin{align*}
  0 &> \frac{\sigma^2}{2}V_{2n+1}''(x_0) +\mu V_{2n+1}'(x_0)-(\delta_1+\lambda_1)V_{2n+1}(x_0) +\lambda_1 V_{2n}(x_0) \\
   &= -\mu-(\delta_1+\lambda_1)V_{2n+1}(x_0) +\lambda_1 V_{2n}(x_0) \\
   &\geq \mu \tilde V'(x_0)- (\delta_1+\lambda_1)\tilde V(x_0) +\lambda_1 V_{2n}(x_0) \\
   &= -\frac{\sigma^2}{2}\tilde V''(x_0) \geq 0\;,
\end{align*}
which is a contradiction. Consequently, we have 
$$ \frac{\sigma^2}{2}V_{2n+1}''(x) +\mu V_{2n+1}'(x)-(\delta_1+\lambda_1)V_{2n+1}(x) +\lambda_1 V_{2n}(x) \geq 0 $$
for any $x\in [0,b_{2n+1})$ which yields the claim.
\end{proof}
Finally, we come to the main statement of this section. Here, we prove that the optimal strategy for the initial control problem is indeed of barrier type.
\begin{thm}\label{t:convergence result}
The sequence $(V_{2n})_{n\in\mathbb N}$ converges together with its first two derivatives locally uniformly to $V(\cdot,\delta_2)$ and its derivatives and the sequence $(V_{2n+1})_{n\in\mathbb N}$ converges together with its first two derivatives locally uniformly to $V(\cdot,\delta_1)$ and its derivatives.
  
In particular, $V(\cdot,\delta_j)$ is a convex, decreasing, positive valued $\mathcal C^2$-function. If $b:=\lim\limits_{n\rightarrow\infty} b_{2n}>0$, then $V''(b,\delta_1)=0$.
The optimal strategy for the initial control problem is the function
$$ Y^*(t) := \sup_{s\in[0,t]} \max\{0,-\inf_{u\in[0,s]}X(u),(b-\inf_{u\in[0,s]}X(u))\one_{\{\ir_s=\delta_1\}}\},\quad t\geq 0.$$
\end{thm}
\begin{proof}
Lemma \ref{l:monotonie} yields that both sequences are monotone increasing and, hence, have a pointwise limit in $[0,\infty]$. Let denote those limits by
\[
U_2(x) := \lim_{n\rightarrow\infty} V_{2n}(x)\quad \mbox{and} \quad U_1(x) := \lim_{n\rightarrow\infty}V_{2n+1}(x)\;.
\]
Since $V_{2n}(x) \leq V^0(x,\delta_2)$ and $V_{2n+1}(x)\leq V^0(x,\delta_2)$ for any $n\in\mathbb N$, $x\geq 0$ we get $U_1(x) \leq V^0(x,\delta_1)<\infty$ and $U_2(x)\leq V^0(x,\delta_2)<\infty$ for $x\geq 0$.
Observe that we have
  \begin{align*}
  |V''_{2n}(x)| &\leq \frac{2}{\sigma^2}\left(\mu |V'_{2n}(x)|+(\delta_2+\lambda_2)|V_{2n}(x)|+\lambda_2|V_{2n-1}(x)|\right) \\
   & \leq \frac{2}{\sigma^2}\left(\mu+(\delta_2+\lambda_2)|V^0(0,\delta_2)|+\lambda_2|V^0(0,\delta_1)|\right).
\end{align*}   
Proposition \ref{p:convergence criterion} yields that the convergence is locally uniformly for the functions and their first derivative. Let $b:=\lim\limits_{n\rightarrow\infty}b_{2n+1}$. Then $V_{2n}$, $V_{2n+1}$ solve (for $n$ large enough) the differential equations \eqref{eq:2n} resp.\ \eqref{eq:2n+1} we conclude that $U_1, U_2$ are $\mathcal C^2$-functions on $(b,\infty)$ and for $x\in(b,\infty)$ we have
  \begin{align*}
     \frac{\sigma^2}{2}U_1''(x) + \mu U_1'(x) -(\lambda_1+\delta_1) U_1(x)+\lambda_1 U_2(x) &= 0 \\
     \frac{\sigma^2}{2}U_2''(x) + \mu U_2'(x) -(\lambda_2+\delta_2) U_2(x)+\lambda_2 U_1(x) &= 0.
  \end{align*}
Since $V_{2n+1}(x)$ are linear on $[0,b]$, we have $V''_{2n+1}(x) = 0 = U''_1(x)$ for $x\in [0,b]$. In particular, $V_{2n+1}$ converges locally uniformly on $\mathbb R_+$ together with its first two derivatives to $U_1$ and its first two derivatives. Thus, the same holds for the convergence of $V_{2n}$ to $U_2$.
  
Finally, Lemma \ref{l:V_n solves HJB} yields for $i,j\in\{1,2\}$ and $i\neq j$ that
\begin{align*}
\min\Big\{\mathcal L_i&(U_i)(x)+\lambda_i U_j(x),U_i'(x)+1\Big\} 
\\&= \lim_{n\rightarrow\infty} \min\Big\{\mathcal L_i(V_{2n+i})(x)+\lambda_i V_{2n+i-1}(x),V_{2n+i}'(x)+1\Big\}=0\;.
\end{align*}
Thus, $(U_1, U_2)$ is the classical solution to the HJB-equation and, hence, $U_1(x) = V(x,\delta_1)$ and $U_2(x)=V(x,\delta_2)$, confer for instance \cite{es} and \cite{HS}. 
\end{proof}
In the following example we illustrate our findings.
\begin{ex}
Consider the following parameters: $\delta_1 := -0.56$, $\delta_2 := 0.1$, $\lambda_1 := 0.57$, $\lambda_2 := 0$, $\mu := 0.05$ and $\sigma :=0.45$. 
\\We have chosen $\lambda_2=0$ for the sake of simplicity.
Consider at first $V^0$, the return function corresponding to the minimal-amount strategy, i.e.\ we apply $Y^0$ in both states.
In the left picture of Figure \ref{fig:1} one sees that the second derivative $(V^0)''(x,\delta_1)$ is negative in some interval close to $0$. In particular, it holds $(V^0)''(0,\delta_1)=-3.3077$. Thus, the strategy $Y^0$ cannot be optimal.
\\Since $\lambda_2=0$, we know that the value function, if starting in the state with $\delta_2>0$, is given by
\[
V(x,\delta_2)=\frac 1{A} e^{-A x}\quad\mbox{with}\quad A=\frac{\mu^2+\sqrt{\mu^2+2\sigma^2\delta_2}}{\sigma^2}\;.
\]
\begin{figure}[t]
\includegraphics[scale=0.3, bb = 0 150 300 650]{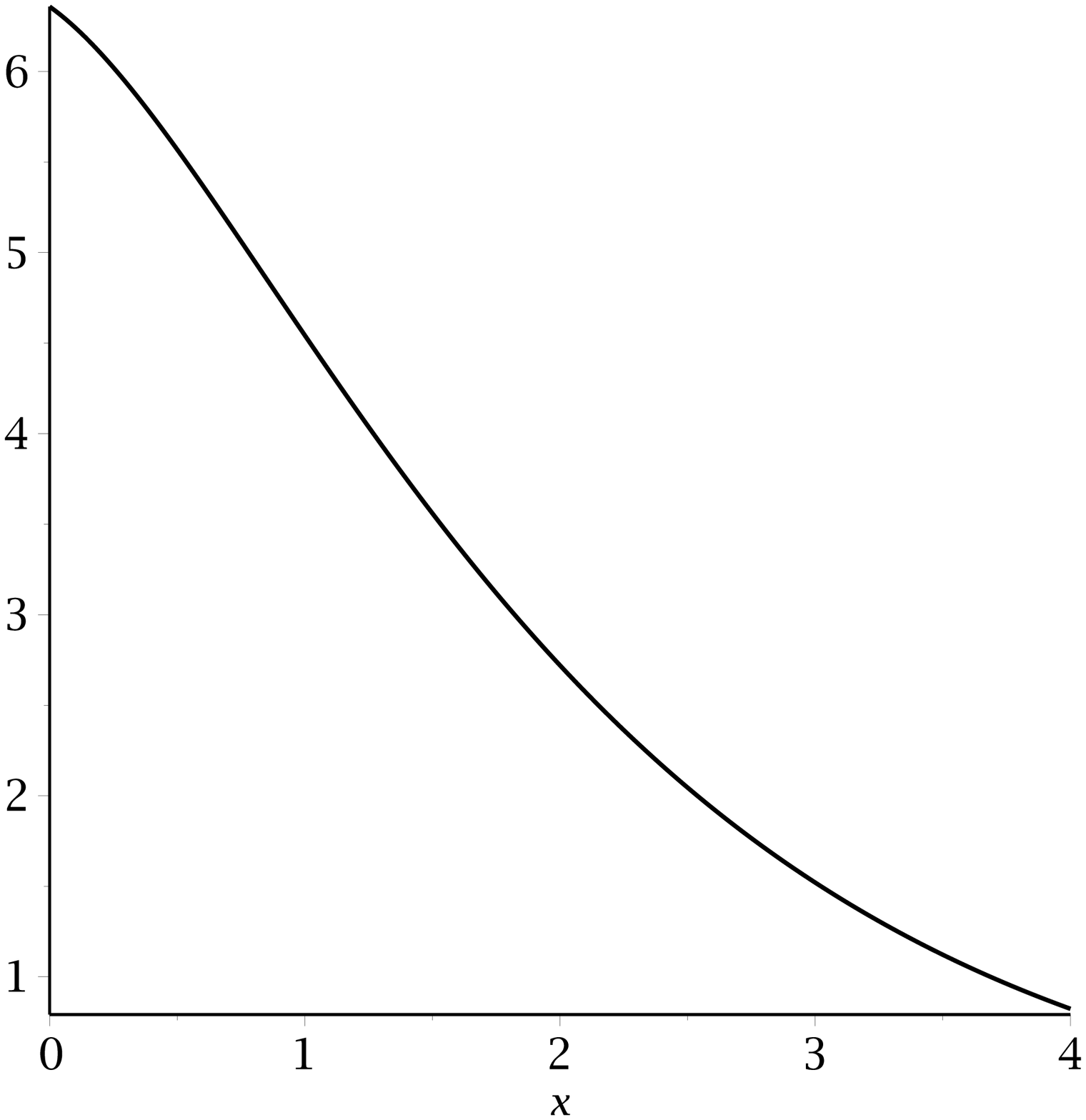}
\includegraphics[scale=0.3, bb = -300 150 0 650]{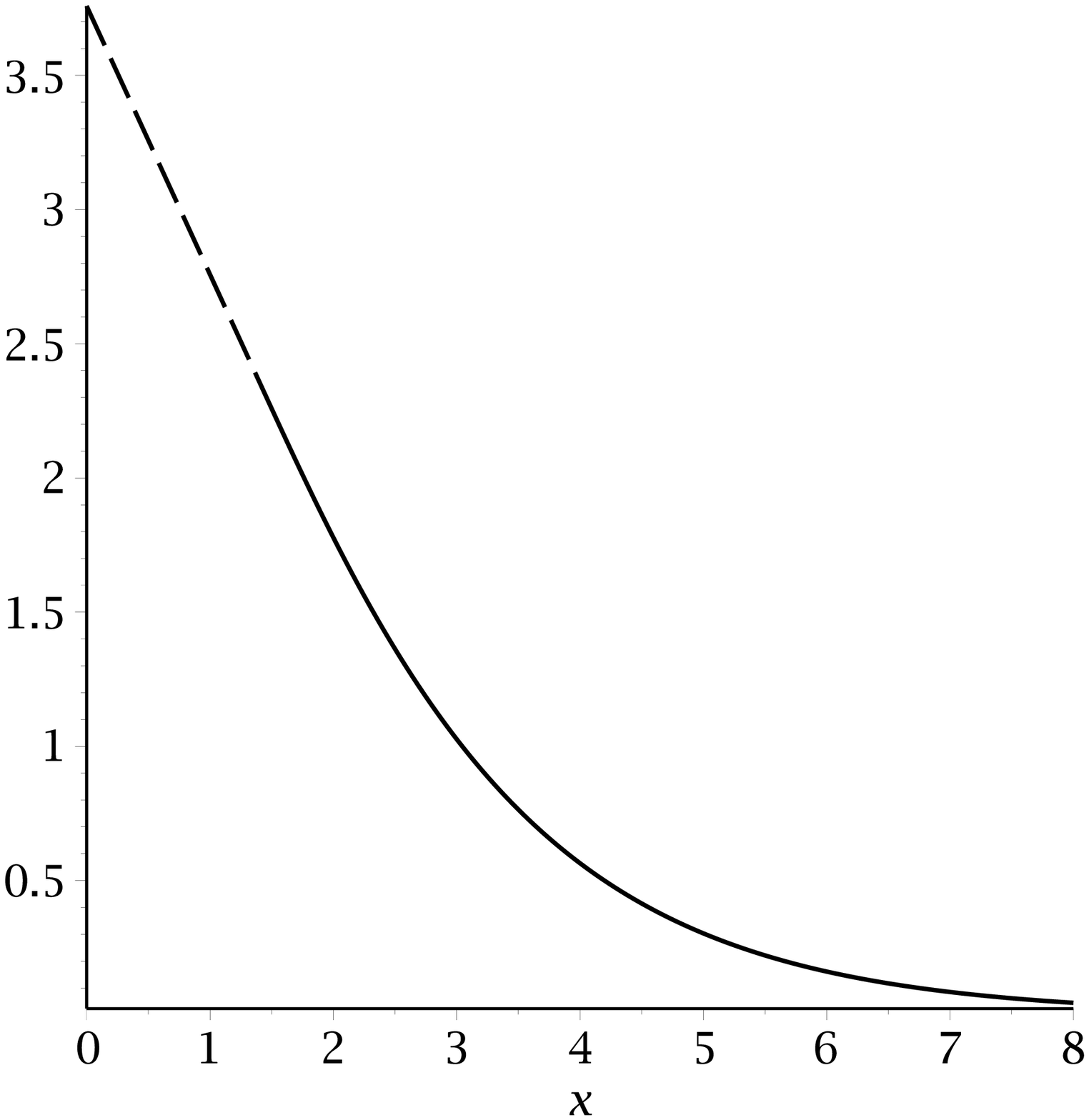}
\caption{The non-convex structure of $V^0(x,\delta_1)$ (left picture) and the value function $V(x,\delta_1)$ (right picture).\label{fig:1}}
\end{figure}
Thus, if we find the optimal barrier, we will be able to calculate the value function via the corresponding differential equation. The optimal constant barrier will minimise the expected discounted capital injections for every $x\in\R_+$. This means, we can choose $x=0$. Denoting the return function corresponding to some barrier $b$ by $V^b$, we obtain with $T_1\sim{ \rm Exp}(\lambda1)$:
\begin{align*}
V^b(0)&=b+V^b(b)=b+\E_0\Big[\int_0^{T_1}e^{-\delta_1 t}\md Y_t^0+e^{(\delta_2-\delta_1) T_1}V\big(b+X_{T_1}^{0},\delta_2\big)\Big]
\\&= b+\E_0\Big[\int_0^{T_1}e^{-\delta_1 t}\md Y_t^0\Big]+\frac{e^{-A b}}{A}\E\Big[e^{(\delta_2-\delta_1) T_1}e^{-A X_{T_1}^{0}}\Big]\;.
\end{align*}
Minimising $V^b(0)$ with respect to $b$, yields the condition
\begin{align*}
1=e^{-A b}\E\Big[e^{(\delta_2-\delta_1) T_1}e^{-A X_{T_1}^{0}}\Big]\;.
\end{align*} 
Since $\lambda_1+\delta_1-\delta_2\neq 0$, we have 
\[
\E\Big[e^{(\delta_2-\delta_1) T_1}e^{-A X_{T_1}^{0}}\Big] = \frac{\lambda_1}{\lambda_1+\delta_1-\delta_2}\frac{\sqrt{\mu^2+2\sigma^2(\lambda_1+\delta_1)}-\sqrt{\mu^2+2\sigma^2\delta_2}}{\mu+\sqrt{\mu^2+2\sigma^2(\lambda_1+\delta_1)}}\;,
\]
confer for instance Borodin and Salminen, \cite[p. 252]{bs}, and the optimal barrier $b^*$ is given by
\begin{align*}
b^*&=\frac 1A\ln\bigg(\frac{\lambda_1}{\lambda_1+\delta_1-\delta_2}\cdot\frac{\sqrt{\mu^2+2\sigma^2(\delta_1+\lambda_1)}-\sqrt{\mu^2 +2\sigma^2\delta_2}}{\mu+\sqrt{\mu^2+2\sigma^2(\delta_1+\lambda_1)}}\bigg)
\\&=1.4248\;.
\end{align*}
Using that $V'(b^*,\delta_1)=-1$ and $V''(b^*,\delta_1)=0$, we can calculate the value function $V(x,\delta_1)$ by solving
\[
\mathcal L_1(f)(x)+\lambda_1 V(x,\delta_2)=0,\quad x\in[b^*,\infty)\;.
\]
In the right picture of Figure \ref{fig:1} one sees $V(x,\delta_1)$, subdivided into the linear part on $[0,b^*]$ and the sum of two exponential functions on $(b^*,8)$.
\end{ex}

\appendix
\appendix
\section{Appendix}
In this section we collect auxiliary mathematical results which might be useful by themselves and are not particularly tight to the topic of the paper. First, we gather properties of a specific second order ODE, its explicit solution under the boundary conditions is given at the beginning of the proof.
\begin{prop}\label{p:ODE}
Let $U:\mathbb R_+\rightarrow [0,\infty)$ be a convex, decreasing and twice continuously differentiable function such that $U$ vanishes at infinity. Let $b,\lambda>0$, $\delta>-\lambda$ and $V$ be the unique solution to the differential equation
    $$ \frac{\sigma^2}{2}V''(x) + \mu V'(x) -(\lambda+\delta)V(x)+\lambda U(x) = 0,\quad x\in[0,\infty)$$
  with $V'(b) = -1$ and $\lim\limits_{x\rightarrow \infty} V(x) = 0$. Then, $V$ is strictly positive valued on $[b,\infty)$, four times continuously differentiable and
\begin{align}
V(b) = \frac 1{A}\Big(1+\frac{2\lambda}{\sigma^2}\int_b^\infty U(y)e^{\tilde A (b-y)}\md y \Big)\label{eq:boundary}
\end{align}
 where 
\begin{align*}
&\psi:=\sqrt{\mu^2+2\sigma^2(\delta+\lambda)} >\mu>0\;,
\\& A=\frac{\mu+\psi}{\sigma^2} \quad\mbox{and}\quad \tilde A:=\frac{-\mu+\psi}{\sigma^2}\;.
\end{align*}
Moreover, $V'$ and $V''$ vanish at infinity. Also, the $J:=\{x\in\mathbb R_+:V''(x)<0\}$ is empty or an interval containing zero and we have $V''\geq 0>V'$ outside of $J$. If $\delta\geq0$ and $b=0$, then $J=\emptyset$.
\end{prop}
\begin{proof}
 We have
 \begin{align*}
&V(x)=C e^{-A (x-b)}+\frac{e^{-A x}}{\psi}\int_b^{x} \lambda U(y)e^{A y}\md y+\frac{e^{\tilde A x}}{\psi}\int_{x}^\infty \lambda U(y)e^{-\tilde A y}\md y\;,
\\&C:=\frac 1{A}\bigg(1+\frac{\tilde A}\psi\int_b^\infty \lambda U(y)e^{\tilde A (b-y)}\md y \bigg)
\end{align*}
for any $x\geq 0$. Since we have $\tilde A>0$, it holds $C>0$. Observe that 
$$V(b) = \frac 1{A}\Big(1+\frac{2\lambda}{\sigma^2}\int_b^\infty U(y)e^{\tilde A (b-y)}\md y \Big)$$
 as required. Consequently, $V(x)>0$ for any $x\in[b,\infty)$. Also $V$ is four times continuously differentiable. Clearly, $V$ and $V'$ vanish at infinity. Inspecting the differential equation yields that $V''$ vanishes at infinity.
 
 $J$ is an open set in $\mathbb R_+$ and, hence, countable union of disjoint open intervals. Let $I\subseteq J$ be one of those open intervals. Define $F(x) := (\lambda+\delta)V'-\lambda U'$ and by taking the derivative on the differential equation we get
  $$ \frac{\sigma^2}{2}V'''(x) = F(x) - \mu V''(x). $$
 $F$ is strictly decreasing on $I$ because $U$ is convex and $V''<0$ on $I$. 
 
 Assume by contradiction that $I$ is non-empty and $a:=\inf(I)>0$. Then $F(a)=\frac{2}{\sigma^2}V'''(a)\leq 0$ and, hence, we have
  $$ \frac{\sigma^2}{2}V'''(x) = F(x) - \mu V''(x) < F(a) - \mu V''(x)\leq -\mu V''(x),\quad x\in I$$
and, hence, $V''$ is strictly decreasing in its zeros of $\overline I$ which implies that $I$ is unbounded and $\lim_{x\rightarrow\infty}V(x)=-\infty$. A contradiction.

 Thus, either $J=\emptyset$ or $0\in J=I$. Also, $J$ is bounded because otherwise $V''<0$ everywhere and, hence, $V'\leq -1$ on $[b,\infty)$ which would imply that $\lim_{x\rightarrow\infty}V(x)=-\infty$. Thus $J$ has the desired structure. Moreover, since $V''\geq 0$ outside $J$ we get $V'$ is increasing outside $J$ and, hence, $V'\leq 0$ outside $J$.
 
 Now assume by contradiction that there is $x\geq\sup(J)$ with $V'(x)=0$. Since $V'$ is increasing and non-positive outside $J$ we get $V'(y) = 0$ for any $y\geq x$ and, hence, $V''(y)=0=V'''(y)$ for any $y\geq x$. Thus, $F(y)=0$ for any $y\geq x$ which implies $U'(y)=0$ for any $y\geq x$. Thus, $U(y)=0$ for any $y\geq x$. Hence, $V(y) = Ce^{-A(y-b)}$ for $y\geq x$ which is a contradiction to $V'(x)=0$.
 
 Consequently, $V'(x)<0$ for any $x\geq\sup(J)$. 
 
 
 Now assume that $\delta\geq 0$, $b=0$ and assume by contradiction that $J\neq \emptyset$. Then $F(0) = -\delta < 0$. Since $F$ is strictly decreasing on $J$ we get $V'''(x) = F(x) -\mu V''(x) < F(0)-\mu V''(x) \leq -\mu V''(x)$. Again, this implies that $V''$ is decreasing around its zeros and, hence $J=\mathbb R_+$. A contradiction.
\end{proof}

\begin{cor}\label{c:2nd.abl.wachs}
  Let $U:\mathbb R_+\rightarrow [0,\infty)$ be a convex, decreasing and twice continuously differentiable function such that $U$ vanishes at infinity. Let $\lambda>0$, $\delta>-\lambda$ and $V_b$ be the unique solution to the differential equation
    $$ \frac{\sigma^2}{2}V_b''(x) + \mu V_b'(x) -(\lambda+\delta)V_b(x)+\lambda U(x) = 0,\quad x\in[0,\infty)$$
  with $V_b'(b) = -1$ and $\lim\limits_{x\rightarrow \infty} V(x) = 0$ and denote $g(b) := V''_b(b)$ for any $b\geq 0$.
  
  Then $g$ is an increasing function.
\end{cor}
\begin{proof}
  We have
  \begin{align*}
    g(b) &= -\frac{2}{\sigma^2}{2}\left(\mu V_b'(b)-(\lambda+\delta)V_b(b)+\lambda U(b)\right) \\
       &= \frac{2\mu}{\sigma^2}+\frac{2\lambda}{\sigma^2}\left(\frac{\lambda+\delta}{\lambda}V_b(b)-U(b)\right) \\
       &= \frac{2\mu}{\sigma^2}+\frac{2\lambda}{\sigma^2}\left(\frac{\lambda+\delta}{A\lambda}\Big(1+\frac{2\lambda}{\sigma^2}\int_b^\infty U(y)e^{\tilde A (b-y)}\md y \Big)-U(b)\right) \\
       &= \frac{2\mu}{\sigma^2}+\frac{2(\lambda+\delta)}{\sigma^2 A}+\frac{2\lambda}{\sigma^2}\left(\frac{2(\lambda+\delta)}{A\sigma^2}\int_0^\infty U(z+b)e^{-z\tilde A }\md y -U(b)\right)
  \end{align*}
  where the third equality is yielded by Proposition \ref{p:ODE} with $A$ and $\tilde A$ given there. Apparently, $g$ is continuously differentiable and we have
  \begin{align*}
    g'(b) &= \frac{2\lambda}{\sigma^2}\left(\frac{2(\lambda+\delta)}{A\sigma^2}\int_0^\infty U'(z+b)e^{-z\tilde A }\md y -U'(b)\right) \\
        &\geq \frac{2\lambda}{\sigma^2}\left(\frac{2(\lambda+\delta)}{A\sigma^2}\int_0^\infty U'(b)e^{-z\tilde A }\md y -U'(b)\right) \\
        &= \frac{2\lambda}{\sigma^2}U'(b)\left(\frac{2(\lambda+\delta)}{A\tilde A\sigma^2} -1\right) = 0
  \end{align*}
  because $A\tilde A = \frac{1}{\sigma^4}(\psi^2-\mu^2) = \frac{2(\lambda+\delta)}{\sigma^2}$. Consequently, $g$ is an increasing function as claimed.
\end{proof}

Sequences of convex $\mathcal C^2$-functions which converge pointwise have very nice convergence behaviour. This observation is our key ingredient for our main result Theorem \ref{t:convergence result} below.
\begin{prop}\label{p:convergence criterion}
Let $(U_n)_{n\in\mathbb N}$ be a sequence of convex $\mathcal C^2$-functions from $\mathbb R_+$ to $\mathbb R$ which converges pointwise to some function $U:\R_+\to \R$ and such that there is $K>0$ with $U_n''(x)\leq K$ for any $x\geq 0$, $n\in\mathbb N$ and assume that $(U_n'(0))_{n\in\mathbb N}$ converges to some $u\in\mathbb R$.
\\Then $U$ is a convex $\mathcal C^1$-function and $U'$ is Lipschitz-continuous with Lipschitz-constant at most $K$. Additionally, $U_n,U_n'$ converge locally uniformly to $U$ resp.\ $U'$.
\end{prop}
\begin{proof}
  Let $t\in[0,1]$ and $x,y\geq 0$. Then we have
   \begin{align*}
      U(tx+(1-t)y) &= \lim_{n\rightarrow\infty}U_n(tx+(1-t)y) \\
                   &\leq \lim_{n\rightarrow\infty} tU_n(x)+(1-t)U_n(y) \\
                   & = tU(x)+(1-t)U(y).
   \end{align*}    
   Thus, $U$ is convex. In particular, $U$ admits a right-derivative on $(0,\infty)$ denoted by $U'$.
   
   Also, we have
   \begin{align*}
      |U'(x) - U'(y)| &= \lim_{h\searrow 0} \left|\frac1h U(x+h)-U(x)-U(y+h)+U(h) \right| \\
       &= \lim_{h\searrow 0}\lim_{n\rightarrow\infty} \left|\frac1h U_n(x+h)-U_n(x)-U_n(y+h)+U_n(h) \right| \\
       &= \lim_{h\searrow 0}\lim_{n\rightarrow\infty} \left|\frac1h \int_0^h U_n'(x+z)-U'_n(y+z)dz \right| \\
       &\leq K|x-y|.
   \end{align*}
Thus, $U'$ is Lipschitz-continuous with constant $K$ on $(0,\infty)$. Consequently, $U$ is $\mathcal C^1$ on $(0,\infty)$ with derivative $U'$. For $h\in(0,x)$ we have
\begin{align*}
  U_n'(x) &\leq \frac1h \int_{x}^{x+h}U_n'(y)\md y 
          = \frac{U_n(x+h)-U_n(x)}{h} 
          \underset{n\rightarrow\infty}\rightarrow \frac{U(x+h)-U(x)}{h}
\end{align*}
and since this is true for any $h$ we get $\limsup\limits_{n\rightarrow\infty}U_n'(x)\leq U'(x)$. Also, we have
\begin{align*}
  U_n'(x) &\geq \frac1h \int_{x-h}^{x}U_n'(y)\md y 
          = \frac{U_n(x-h)-U_n(x)}{h} 
          \underset{n\rightarrow\infty}\rightarrow \frac{U(x-h)-U(x)}{h}
\end{align*}
 and, hence, $\liminf\limits_{n\rightarrow\infty}U_n'(x)\ge U'(x)$. Consequently, $U_n'(x)\rightarrow U'(x)$ for any $x>0$.
 
 We have 
 \begin{align*}
   |U_n'(x)| \leq Kx + |U_n'(0)|
             \leq Kx + \sup_{n\in\mathbb N}U_n'(0)
 \end{align*}
 and, hence, the dominated convergence theorem yields that $U_n \rightarrow U$ locally uniformly.
 
 Now, let $K>0$ and $\epsilon >0$. We will show that there is $N\in\N$ such that that $\sup\limits_{n\geq N}\sup\limits_{x\in[0,K]} |U_n'(x)-U'(x)|\leq \epsilon$ which yields that $U_n'\rightarrow U'$ locally uniformly.
 
 For that choose $N\in\mathbb N$ such that 
  $$\sup_{n\geq N}\sup_{x\in[0,2K]} |U_n(x)-U(x)|\leq 1\wedge \left(\frac{\epsilon}{2K+2}\right)^2=:\delta.$$ 
 Then we have for $n\geq N$ and $x\in[0,K]$ with $h:=\sqrt{\delta}$
  \begin{align*}
    |U_n'(x)-U'(x)| &\leq \frac1h \int_x^{x+h}U_n'(y)dy - U_n'(x) + \left|\frac1h \int_x^{x+h}U_n'(y)dy -U'(x)\right | \\
     &\leq Kh + \left| \frac{U_n(x+h)-U_n(x)}{h} -U'(x)\right| \\
     &\leq Kh + 2\delta/h + \left|\frac{U(x+h)-U(x)}{h} -U'(x)\right| \\
     &\leq 2Kh + 2\delta/h \\
     &= \epsilon\wedge (2K+2) \leq \epsilon.
  \end{align*}
  Since the estimate is independent of $n$ and $x$ we get the required convergence.
\end{proof}


\end{document}